\documentclass[aps,showpacs,amsmath,amssymb,onecolumn,superscriptaddress,notitlepage,longbibliography,preprintnumbers,prx]{revtex4-2}

\usepackage{float}
\usepackage{comment}
\usepackage[T1]{fontenc}
\usepackage[utf8]{inputenc}
\usepackage{lmodern}
\usepackage{amsmath,amssymb,amsthm,mathtools,bm,mathrsfs}
\usepackage{bbm}
\usepackage{microtype}
\usepackage{xcolor}
\usepackage{enumitem}
\usepackage{graphicx}
\usepackage[colorlinks=true, allcolors=blue]{hyperref}
\usepackage[table, dvipsnames]{xcolor} 
\usepackage{comment}

\usepackage{subcaption}
\usepackage{cleveref}
\usepackage{dsfont}
\usepackage{multirow}
\usepackage{booktabs}
\usepackage{orcidlink}
\usepackage[normalem]{ulem}
\usepackage{ragged2e}
\usepackage{calc}
\usepackage{algpseudocode}
\algrenewcommand\algorithmicrequire{\textbf{Input:}}
\algrenewcommand\algorithmicensure{\textbf{Output:}}

\makeatletter
\newcounter{algorithm}

\newenvironment{algorithm}[1][t]
  {%
   \let\c@figure\c@algorithm
   
   \long\def\@makecaption##1##2{%
     \par\noindent{\bfseries ##1}\ ##2\par\nobreak\vspace{3pt}\hrule\vspace{4pt}}%
   \begin{figure}[#1]\hrule\vspace{4pt}\raggedright}
  {\par\vspace{2pt}\hrule\end{figure}}
\makeatother

\newcommand{\ket}[1]{\left|#1\right\rangle}
\newcommand{\bra}[1]{\left\langle#1\right|}
\newcommand{\braket}[2]{\left\langle#1|#2\right\rangle}

\newtheorem{proposition}{Proposition}

\begin{document}

\title{Q-MERGE: Parallelising Quantum State Preparation for Large-Scale Classical Data}
\author{Archie Butterworth}
\affiliation{Centre for Quantum Information, 	
Simulation and Algorithms, \\ 
The University of Western Australia, Perth, Australia}
\author{Jens Renders}
\affiliation{Centre for Quantum Information, 	
Simulation and Algorithms, \\ 
The University of Western Australia, Perth, Australia}
\affiliation{University Paris City and University of Reunion, France}
\author{Jingbo Wang}
\email{jingbo.wang@uwa.edu.au}
\affiliation{Centre for Quantum Information, 	
Simulation and Algorithms, \\ 
The University of Western Australia, Perth, Australia}


\begin{abstract}
Quantum processing of classical data fundamentally relies on efficiently mapping classical datasets onto the amplitudes of quantum states. Preparing large amplitude-encoded states, however, remains a major bottleneck in quantum computing. In this paper, we introduce a novel state-preparation framework we call Q-MERGE, that addresses this bottleneck by partitioning a large target state into $M$ $n$-qubit segments, preparing these segments independently and in parallel, and then coherently combining them into a single amplitude-encoded state using SELECT-SWAP operations and measurement. Q-MERGE is agnostic to the segment-level preparation method, allowing existing techniques to be applied to smaller subproblems while providing a tunable trade-off between circuit depth and qubit count. With mid-circuit measurement and preparation-register reuse, the required ancilla qubits can be reduced from $\mathcal{O}(Mn)$ to $\mathcal{O}(n)$. Applied to a real-world \(32\times1024\) ultrasound dataset, Q-MERGE achieves an infidelity of $3.5 \times10^{-9}$, compared with $6.2 \times 10^{-3}$ for direct preparation using the same underlying method, a more than six orders of magnitude improvement at equal CNOT circuit depth. We demonstrate its feasibility experimentally with a 53 qubit circuit comprised of 954 1-qubit gates and 478 2-qubit gates, on the Quantinuum System Model H2 trapped-ion quantum computer. The prepared state is validated using shadow-overlap tomography. Numerical analysis of Haar-random states indicates consistent success probability for up to $M=10^7$ segments, supporting the scalability of Q-MERGE for encoding massive classical datasets.

\end{abstract}

\maketitle
\section{Introduction}
While the current quantum computing literature contains a remarkable set of algorithms promising quantum advantage over classical algorithms \cite{montanaro2016quantum}, often claiming polylogarithmic run time with respect to the input dimension \cite{harrow2009quantum,lloyd2014quantum,kerenidis2017quantum}, few maintain this advantage without assuming that the input data is efficiently accessible as an amplitude encoded quantum state \cite{aaronson2015read, ciliberto2018quantum, tang2021quantum}. In the worst case, preparing quantum states that represent classical data in its probability amplitudes has exponential elementary gate complexity in the number of qubits \cite{yuan2023optimal, plesch2011quantum,Mottonen2004-mk}, which dominates the end-to-end complexity of these algorithms. This highlights the importance of domain specific optimisation of this bottleneck. Current state preparation techniques aim to exploit sparsity, low-rank, or other structural properties of the input data \cite{gleinig2021efficient, araujo2021low, zylberman2024efficient, gonzalez2024efficient,Grover2002-kj}.

In recent years, tensor networks \cite{orus2014practical, bridgeman2017hand} have been shown to be a promising candidate for constructing state preparation methods. In particular, many methods use the embedding of a tensor network in a quantum circuit as an ansatz in which the parameters describing this circuit can be optimised over, in order to achieve high-fidelity embeddings of a tensor network in a quantum circuit \cite{ran2020encoding, melnikov2023quantum, rudolph2024decomposition, green2025quantum, green2026quantum}. These methods have been shown to be effective in preparing quantum states which display 1D entanglement structure.

However, for large datasets requiring many qubits, or for data containing complex entanglement structure, these optimisation based tensor network approaches to state preparation often struggle to prepare states to high fidelity. The optimisation landscape becomes increasingly complex, and the optimisation over the circuit ansatz will often become trapped in local minima, causing the fidelity to plateau, preventing high-fidelity encodings.

This is in line with the general hardness of quantum circuit optimisation: in general, reducing a state preparation circuit to a sequence of gates with minimal depth, from a given gate set, can be seen as a form of the NP-hard minimum-length generator sequence problem \cite{even1981minimum}. Other relevant optimisation metrics such as T-count and T-depth are also known to be NP-hard in the universal Clifford+T gate set \cite{vandewetering2024optimisingquantumcircuitsgenerally}. These negative results motivate approaches that reduce the problem size by cutting the data into several smaller segments, to be prepared separately. For 2D image data in particular, choosing the rows or columns as the segments not only reduces the size, but also reduces the 2D structure to 1D, enabling the application of tensor networks designed to capture 1D structure.

However, preparing each segment separately on separate registers leads to an encoding with suboptimal space utilisation and reduced entanglement structure, compared to the direct amplitude encoding of the full state, and it discards the relative scale of each segment. On the other hand, preparing the state segment by segment on one register prohibits the parallel preparation of each segment and requires adding multiple control qubits to each of the state preparation circuits. This \textit{controlisation} multiplies the gate count of the preparation circuits, or is even impossible to achieve exactly if no explicit circuit decomposition for the preparation unitaries is known \cite{araujo2014quantum, nakayama2015quantum}.

To overcome these problems, we propose a quantum algorithm called Q-MERGE (Quantum Multi-ReGister Encoding) that compresses a quantum state prepared on separate registers onto a target register, with flexible trade-off between circuit depth and circuit width. If enough qubits are available, our algorithm allows all data segments to be prepared in parallel and then combined into the same state one would have obtained by direct amplitude encoding. At the other extreme, when the available qubits can accommodate only a fixed number of data segments but support mid-circuit measurement, we present a modified but mathematically equivalent protocol that achieves the same compression. This case is of particular interest when the classical segments can be generated on the fly (either by an algorithm, or as a result of measurements), as this allows the encoding of classical data of exponential size, with respect to the size of both the classical and quantum computers involved.



\section{Methods}
A target state $\ket{\psi}$ can be broken down into a set of equally sized, normalised segments $|\psi_i\rangle \in \mathbb{C}^N, i\in \{0,\cdots,M-1\}$, such that 
\begin{equation}
    \ket{\psi}=\sum_{i=0}^{M-1}\sigma_i  \ket{\psi_i} \otimes \ket{i} \in \mathbb{C}^{NM}
\end{equation}
where the weights $\sigma_i \in \mathbb{R}^+$ account for the relative magnitude of each segment in the original state $\ket{\psi}$. Once an efficient state preparation method $U_{\mathrm{prep}}^{(i)}:|0\rangle \mapsto |\psi_i\rangle$ is obtained for each segment of data, it can be applied directly on separate quantum registers to prepare the state
\begin{equation}
|\psi_\text{tensor}\rangle =|\psi_0\rangle \otimes |\psi_1\rangle \otimes \cdots \otimes |\psi_{M-1}\rangle \in \mathbb{C}^{N^M}.
\end{equation}
Without entanglement between the registers, the space utilisation of this state is suboptimal. It occupies a Hilbert space of dimension $N^M$, which is much larger than the dimensionality of the data: $NM$, and requires $M\lceil\log(N)\rceil$ qubits instead of the optimal $\lceil\log(NM)\rceil$. Moreover, $|\psi_\text{tensor}\rangle$ is a non-linear function of the state vectors $|\psi_i\rangle$, which may hinder downstream tasks that benefit from linearity, such as quantum machine learning and quantum linear algebra algorithms.

The first stage of the Q-MERGE algorithm presented in this paper transforms the easily prepared state $|\psi_\text{tensor}\rangle$, with guaranteed success, into the weighted stack
\begin{equation}
|\psi_\text{w}\rangle = \frac{1}{\sqrt{\mathcal{N}}}\sum_{i=0}^{M-1} \tilde{w}_i  |\psi_i\rangle \otimes \ket{i} \in \mathbb{C}^{NM}, \quad\mathcal{N}=\sum_{i=0}^{M-1}|\tilde{w}_i|^2
\end{equation}
where $\tilde{w}_i \in \mathbb{C}$ are returned as a classical output of the algorithm. Like $|\psi\rangle$ itself, the state $|\psi_\text{w}\rangle$ requires only $\lceil\log(NM)\rceil$ qubits and is linear in the state vectors $|\psi_i\rangle$. Our approach avoids the use of sequential controlled $U_{\mathrm{prep}}^{(i)}$ operations, instead applying all $U_{\mathrm{prep}}^{(i)}$ in parallel, using only $\mathcal{O}(M\log{N})$ ancilla qubits, which can be reduced to $\mathcal{O}(\log{N})$ by trading in parallelism, providing a highly flexible and scalable approach to state preparation.

Additionally, the second stage of Q-MERGE uses the classically known weights $w_i$  to recover the original state $\ket{\psi}$, through the application of a block-encoded diagonal matrix. We analyse the success probability of post-selecting the desired block, which depends on the distribution of the weights, which in turn depend on the distribution of the encoded data. In case of data with high variance in the amplitudes, this success probability can become prohibitively small. This issue is overcome by a classical preconditioning procedure, moving this variance into the phases, at the cost of a shallow layer of 2 qubit gates in the quantum circuit.

Note that the derivation of Q-MERGE below assumes that the segment level preparation circuits are exact. If instead, the preparation circuits $U_\text{prep}^{(i)}$ prepare approximations $\ket{\tilde{\psi}_i}$ of the segments, Q-MERGE will prepare
\begin{equation}
    \ket{\tilde \psi} = \sum_{i=0}^{M-1}\sigma_i\ket{\tilde{\psi_i}}\otimes\ket{i} .
\end{equation}
The global phase of each $U_\text{prep}^{(i)}$ can be chosen such that $\braket{\tilde{\psi}_i}{\psi_i} \in \mathbb{R}^+$, and in that case the fidelity of $\ket{\tilde\psi}$ to $\ket{\psi}$ can be expressed in terms the segment level fidelity $F_i = |\braket{\tilde\psi_i}{\psi_i}|^2$ as
\begin{equation}
  F = |\braket{\tilde\psi}{\psi}|^2 = \left|\sum_{i=0}^{M-1}\sigma_i^2 \braket{\tilde\psi_i}{\psi_i}\right|^2=  \left(\sum_{i=0}^{M-1}\sigma_i^2 \sqrt{F_i}\right)^2 .
  \label{eq:multiregister fidelity}
\end{equation}

\subsection{Q-MERGE Stage 1: Preparation and compression}
\label{sec:stage1}

\begin{algorithm}[H]
\caption{Q-MERGE Stage 1: parallel preparation and SELECT-SWAP compression, using one preparation
register per segment (Fig.~\ref{fig:naivecircuit}). The
partition of $\ket{\psi}$ and th e synthesis of each $U_{\mathrm{prep}}^{(i)}$ are
classical preprocessing, carried out beforehand.}
\label{alg:stage1}
\begin{algorithmic}[1]
\Require Normalised segments $\{\ket{\psi_i}\in\mathbb{C}^{N}\}_{i=0}^{M-1}$, zero-padded
         to length $N$ if required, from the partition
         $\ket{\psi}=\sum_i\sigma_i\ket{\psi_i}\otimes\ket{i}$ with $M=2^m$, $N=2^n$;
         preparation circuits
         $U_{\mathrm{prep}}^{(i)}:\ket{0}^{\otimes n}\mapsto\ket{\psi_i}$.
\Ensure  Weighted stack $\ket{\psi_\textbf{w}}$ on registers $SI$, and the measurement
         record $\vec{u}=(u_0,\dots,u_{M-1})$.
\Statex
\State Initialise $S\gets\ket{+}^{\otimes n}$ and $I\gets\ket{+}^{\otimes m}$.
\For{$i=0$ \textbf{to} $M-1$ \textbf{in parallel}}
    \State Apply $U_{\mathrm{prep}}^{(i)}$ to preparation register $A_i$.
\EndFor
\For{$i=0$ \textbf{to} $M-1$}
    \State Apply SELECT-SWAP between $A_i$ and $S$, controlled on $\ket{i}_I$.
\EndFor
\State Measure every $A_i$ in the computational basis and record the outcomes $\vec{u}$.
\State \Return $\ket{\psi_\textbf{w}}
       = \frac{1}{\sqrt{\mathcal{N}}}\sum_{i=0}^{M-1}\tilde{w}_i\ket{\psi_i}_S\ket{i}_I$
       and $\vec{u}$, where $\tilde{w}_i=\prod_{k\neq i}(\psi_k)_{u_k}$ and
       $\mathcal{N}=\sum_i|\tilde{w}_i|^2$.
\end{algorithmic}
\end{algorithm}
To encode a target state $\ket{\psi} \in \mathbb{C}^{MN}$, it is partitioned into $M$ smaller segments as
\begin{equation}
    \ket{\psi}=\sum_{i=0}^{M-1}\sigma_i  \ket{\psi_i} \otimes \ket{i}  \qquad \ket{\psi_i}\in \mathbb{C}^N \qquad \sigma_i\in\mathbb{R}^+
\end{equation}
The choice of partitioning should be made such that each $\ket{\psi_i}$ has a favourable state preparation circuit $U_{\mathrm{prep}}^{(i)}$, with consideration given to the number of qubits and therefore the number of partitions $M$ that the hardware allows. Given these $M = 2^m$ state vectors $\ket{\psi_i}=\sum_{k=0}^{N-1}(\psi_i)_k \ket{k}\in \mathbb{C}^N$ with $N = 2^n$,
the protocol uses the following resources:
\begin{center}
\begin{tabular}{lll}
\toprule
Register & Qubits & Role \\
\midrule
Preparation registers \{$A_0, \dots, A_{M-1}$\}  & $Mn$ & Prepare and hold each $\ket{\psi_i}$\\
Storage $S$ & $n$ & The register in which each vector is loaded/swapped into\\
Index $I$ & $m = \log_2 M$ & Index each loaded vector\\
\bottomrule
\end{tabular}
\end{center}

The initial states $\ket{\psi_i}$ are prepared on each of the preparation registers using their respective $\{U_{\mathrm{prep}}^{(i)}\}_{i=0}^{M-1}$ circuits, and the storage $S$ and index register $I$ are prepared in equal superposition of all states:
\begin{equation}
  |\psi'\rangle
  = \bigotimes_{i=0}^{M-1} \ket{\psi_i}_{A_i}
  \otimes |+\rangle^{\otimes n}_S \otimes |+\rangle^{\otimes m}_I.
\end{equation}
We then apply SELECT-SWAP unitary gates
$\sum_i |i\rangle\langle i|_I \otimes \mathrm{SWAP}(A_i, S)$
 to enact swaps between the storage register $S$ and the preparation register $A_i$, controlled on the index register state $\ket{i}$. This creates the state:
\begin{equation}
  |\psi''\rangle
  = \frac{1}{\sqrt{M}}
    \sum_{i=0}^{M-1}
      \ket{\tilde{w}_i}
      |\psi_i\rangle_S |i\rangle_I, \qquad \ket{\tilde{w}_i}=|+\rangle^{\otimes n}_{A_i}
      \Big(\bigotimes_{k \neq i} |\psi_k\rangle_{A_k}\Big) 
  \label{eq:entangled-state}
\end{equation}
Measuring all $M$ preparation registers in the computational basis, and obtaining outcome $\vec{u} = (u_0, \dots, u_{M-1})$ , $u_i \in \{0, 1,\dots ,N-1\}$, collapses the state onto the storage and index register resulting in the state
\begin{equation}
    \ket{\psi_\text{w}}= \frac{1}{\sqrt{\mathcal{N}}}
    \sum_{i=0}^{M-1} \tilde{w}_i \, \ket{\psi_i}_S |i\rangle_I,
  \qquad
  \tilde{w}_i = \prod_{k \neq i} (\psi_k)_{u_k},
  \qquad
  \mathcal{N} = \sum_{i=0}^{M-1} |\tilde{w}_i|^2.
\end{equation}
Expressing in terms of the normalised weights $w$ gives
\begin{equation}
    \ket{\psi_\text{w}}= 
    \sum_{i=0}^{M-1} w_i \, \ket{\psi_i}_S |i\rangle_I,
  \qquad
  w_i = \frac{\tilde{w}_i}{\sqrt{\mathcal{N}}}\label{eq:output_after_measure_un}
\end{equation}
which we refer to as the ``weighted stack state''.
This process (Q-MERGE Stage 1) of preparing each individual $\ket{\psi_i}$ onto the preparation registers $A_i$, and subsequently swapping them into the storage register $S$, and finally measuring on all preparation registers to generate $\ket{\psi_\textbf{w}}$ is shown in Fig. \ref{fig:naivecircuit}.
\begin{figure}
    \centering
    \includegraphics[]{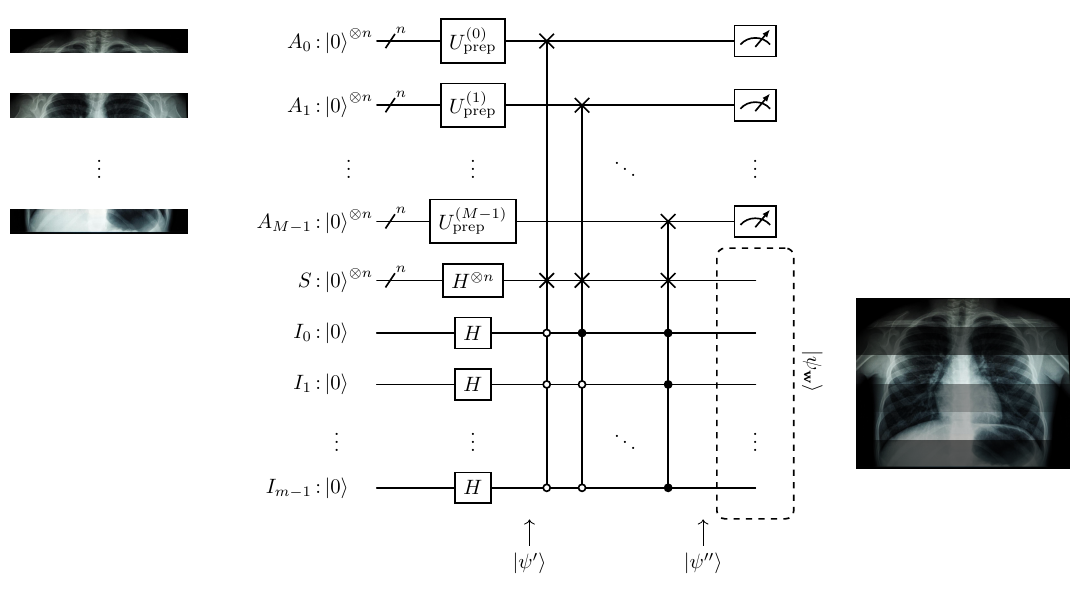}
    \caption{Q-MERGE Stage 1 circuit with $M$ preparation registers. The chest x-ray illustrates the segments prepared on the preparation registers and the resulting weighted stack on the storage and index register. The random weights can be removed by stage 2 of Q-MERGE}
    \label{fig:naivecircuit}
\end{figure}

The coefficient of branch $k$ is the product of the $u_k$-th entries of all \emph{other} input vectors. The $\ket{\psi_\text{w}}$ state contains all the information of the state vector encoding, up to relative phase and scale of the segments. Depending on the application, the weights $w_i$ can be taken into account or ignored in downstream tasks, ending the preparation procedure here. For applications where this is not feasible, the weights can be corrected by stage 2 of Q-MERGE.

\subsection{Q-MERGE Stage 2: Weighting Redistribution}
\begin{algorithm}[ht]
\caption{Q-MERGE Stage 2: weighting redistribution. Restores the original segment weights
$\sigma_i$, succeeding with probability $1/\beta^2$, which is known from Algorithm \ref{alg:stage1}.}
\label{alg:stage2}
\begin{algorithmic}[1]
\Require $\ket{\psi_\textbf{w}}$ and the record $\vec{u}$ from Algorithm~\ref{alg:stage1};
         target weights $\{\sigma_i\}$; flattening circuit $V$, or $V=I_{SI}$ if no
         classical preconditioning was used.
\Ensure  The target state $\ket{\psi}$ on $S\otimes I$.
\Statex
\State Compute $w_i$,
       $\alpha_i=\,\sigma_i/w_i$ and $\beta=\max_i|\alpha_i|$
       classically.
\State Append an ancilla qubit $Q$ in $\ket{0}$.
\State Apply a $(\beta,1,0)$-block encoding of the diagonal reweighting operator
       $D=\sum_{i=0}^{M-1}\alpha_i\,I_S\ket{i}\bra{i}$ to $S\otimes I$, with
       $Q$ as the block-encoding ancilla.
\State Measure $Q$.
\If{the outcome is not $\ket{0}$}
    \State Discard and restart from Algorithm~\ref{alg:stage1}.
\EndIf
\State Apply $V^{\dagger}$ to $S\otimes I$.
       
\State \Return $\ket{\psi}=\sum_{i=0}^{M-1}\sigma_i\ket{\psi_i}_S\ket{i}_I$.
\end{algorithmic}
\end{algorithm}
The first stage of Q-MERGE is able to prepare the weighted stack
\begin{equation}
  \ket{\psi_\text{w}}= \sum_{i=0}^{M-1}w_i  \ket{\psi_i}_S \otimes \ket{i}_I
\end{equation}
with guaranteed success. However, for applications where this state is not a useful encoding of the data, and the unmodified state
\begin{equation}
  \ket{\psi}=\sum_{i=0}^{M-1}\sigma_i \ket{\psi_i} \otimes \ket{i}
\end{equation}
with the original weighting needs to be prepared, we present a process to restore each $\sigma_i$ to the segments, present in the original data, from the weighted stack state.

All weights $w_i$ are fully determined by the classical measurement $\vec{u}$ and the known input data, so the required corrections to the weights are classically computable once the measurements of \textit{all} preparation registers have been made. This reweighting map on register $SI$ can be expressed as a diagonal operator:
\begin{equation}
    D = \sum_{i=0}^{M-1}\alpha_iI_S\ket{i}\bra{i},
    \qquad
    \frac{D\ket{\psi_\text{w}}}{\lVert D\ket{\psi_\text{w}}\rVert} = \ket{\psi}
\end{equation}
with $\alpha_i = \frac{\sigma_i}{w_i}$. A diagonal operator can be efficiently implemented in a quantum circuit as a $(\beta, 1, 0)$-block encoding, with $\beta = \max \{|\alpha_i|:i=0,\cdots,M-1\}$ through index controlled rotations on a single ancilla qubit $Q$ \cite{takahira2021quantum}. In this case, $M$ controlled rotations are sufficient. The block encoding is described in more detail in Sec. \ref{sec:block encoding}. The post-selection success probability of selecting the desired block is given by
\begin{equation}
    P\left(\ket{0}_Q\right) = \left\lVert \frac{D}{\beta} \ket{\psi_\text{w}} \right\rVert^2 = \left\lVert \frac{1}{\beta}\ket{\psi} \right\rVert^2 = \frac{1}{\beta^2}.
\end{equation}

\subsection{Optimising success probability via classical preconditioning} 
\label{sec:classical preconditioning}
The success probability of Q-MERGE is determined by $\beta$, the magnitude of the reweighting. If $\ket{\psi}$ contains many small amplitudes, the sampling of a small $w_i$ is likely. If additionally, $\ket{\psi_i}$ still has a large contribution $\sigma_i$ in $\ket{\psi}$, this requires a large reweighting, which negatively influences $\beta$ and the success probability.

A key observation leading to the resolution of this issue is that the action of only a few \textit{brickwall} layers of SU(4) gates acting on a state can dramatically change the amplitudes of a state vector, where the majority of the complex entanglement structure remains stored in the phases of the basis states rather than the amplitudes. Our past work \cite{Butterworth2026-rb} has shown that a quantity characterising the distribution of probability amplitudes, known as the inverse participation ratio (IPR), can be dramatically minimised in this efficient way, which will allow us to boost the success probability as described below. Once this ``flattened'' state is prepared, we can easily apply our shallow flattening procedure in reverse to restore the desired state. 

For a normalised state $\ket{\psi}=\sum_{k=0}^{D-1}(\psi)_k\ket{k}\in\mathbb{C}^D$, the inverse participation ratio (IPR) is defined as
\begin{equation}
    \mathrm{IPR}(\psi) = \sum_{k=0}^{D-1} \bigl|(\psi)_k\bigr|^4,
    \label{eq:ipr}
\end{equation}
which is bounded as $1/D \leq \mathrm{IPR}(\psi) \leq 1$: the lower bound is saturated by the uniform state $(\phi)_k=1/\sqrt{D}$, and the upper bound by any computational basis state. A smaller $\mathrm{IPR}$ therefore corresponds to a flatter amplitude distribution. Given that we wish to prepare our target state $\ket{\psi}$  on an $(m+n)$-qubit register $S\otimes I$ of dimension $MN = 2^{m+n}$, we can form a shallow ``flattening circuit'' $V\in SU(MN)$ that will change the possibly ill-conditioned $\ket{\psi}$ into a state $V\ket{\psi}=\ket{\tilde{\psi}}$ with a lower IPR and therefore a more favourable expected reweighting success probability.  Sec. \ref{sec:minimum IPR} shows that a state with minimal IPR has a reweighting success probability of exactly 1, motivating the choice of IPR as a proxy for the true reweighting success. The state $\ket{\tilde{\psi}}$ can then be prepared with this high success probability, followed by an additional final step of applying $V^{\dagger}$ to restore $\ket{\psi}$. 

We take $V$ to be $L$ layers of \textit{brickwall} two-qubit SU(4) gates acting on the $S \otimes I$ register. Labelling the $m+n$ qubits $q_0,\dots,q_{m+n-1}$, the $L$-layers of the brickwall circuit is
\begin{equation}
    V =
    \prod_{\ell=0}^{L-1}
    \left[
        \prod_{\substack{j=0\\j\equiv \ell\;(\mathrm{mod}\;2)}}^{m+n-2}
        V_{\ell,j}(q_j,q_{j+1})
    \right],
    \qquad
    V_{\ell,j}\in SU(4),
    \label{eq:brickwall}
\end{equation}
where each layer consists of two-qubit gates acting on alternating nearest-neighbour pairs. The gates within each layer act on disjoint qubit pairs and can therefore be applied in parallel, giving a circuit depth of $L$. Each $V_{\ell,j}$ is parametrised by $15$ real parameters (the generators of $\mathfrak{su}(4)$), which are optimised classically to minimise $\mathrm{IPR}\bigl(V\ket{\psi}\bigr)$ as in Eq.~\eqref{eq:ipr}. Since $V$ is unitary by construction, the same staircase structure in reverse with each $V_{\ell,j}$ replaced by $V_{\ell,j}^\dagger$ exactly undoes the transformation once the flattened state $\ket{\tilde{\psi}}$ has been prepared on separate registers, undergone the SELECT-SWAP and measurement process and reweighted, completely bypassing the difficulty of reweighting $\ket{\psi}$.
\section{Complexity Analysis}
The analysis of the Q-MERGE framework presented in Fig. \ref{fig:naivecircuit} assumes that each segment of the data can be efficiently prepared, which will always be possible by taking $M$ large enough. The validity of this claim is supported in Sec. \ref{sec:numerical results}, where orders-of-magnitude improvement over direct encoding is observed.

\subsection{Parallel Qubit Circuit Design}

The depth of the circuit on the main register (storage $S$ and index $I$), is purely controlled by the depth induced by the multi controlled SELECT-SWAP gates, with SELECT-SWAP denoting a SWAP between two registers of size $n$, conditional on $m$ qubits. $M$ of these SELECT-SWAP operations are required, one for each register $A_i$ that needs to be loaded into the storage register $S$. As noted in Sec. \ref{sec:swapdecomp}, each of these can be decomposed into SWAP and $\mathrm{C}^{m}\mathrm{NOT}$ gates both of which efficient decompositions exist, \cite{Claudon2024-ry}, leading to a total circuit depth of $\mathcal{O}(M\log^3{(m)}+Mn)$ for the whole process on the main registers.

The depth of each preparation register does not impact the depth of the circuit on the main register. One should prepare the state $\ket{\psi_i}$ on the preparation register $A_i$ immediately before the SELECT-SWAP to the main register occurs, meaning that on each preparation register the depth is only as deep as $U_{\mathrm{prep}}^{(i)}$, with $\mathcal{O}(n)$ CSWAP gates as outlined in Sec. \ref{sec:swapdecomp}. 

Before Q-MERGE acts on the initial state prepared by the  $U_{\mathrm{prep}}^{(i)}$, we are storing the state information on $Mn + m + n$ qubits ($Mn$ for the preparation registers, $m$ for the index register and $n$ for the storage register), leading to a $N^M N M$ total dimension of the full Hilbert space, but only $MN$ dimensions of this are used to store the information of the desired state. After the compression we can discard the preparation registers to leave $n + m$ qubits with the resulting state stored on the full $NM$ available dimensions in the Hilbert space of the reduced system.

The framework described uses $\mathcal{O}(Mn)$ qubits, a trade-off between circuit depth and number of qubits. By allowing mid-circuit measurement and the reuse of the preparation registers, we can reduce this to $\mathcal{O}(n)$ number of qubits needed, while keeping the $\mathcal{O}(M\log^3{(m)}+Mn)$ circuit depth and equivalent circuit architecture.

\subsection{Compact Qubit Circuit Design}
\label{sec:alt design}
Allowing for mid-circuit measurement and re-use of the preparation registers can reduce the total qubit count from $\mathcal{O}(Mn)$ to $\mathcal{O}(n)$. Essentially, a fixed number of preparation registers are chosen, say $\tilde{M} < M$, which are reused after each measurement. This follows from the observation that, in the above framework, the order in which the state preparation and measurement occurs does not matter: measuring a preparation register immediately after it has been acted on by the SELECT-SWAP is equivalent to measuring all registers at the end of the circuit. Therefore the register can be reused, if hardware allows, to perform the preparation of subsequent segments once reset after measurement.

The alternative circuit design is equivalent to Fig.~\ref{fig:naivecircuit}, except it has $\tilde{M}$ as a constant number of preparation registers used to encode $M$ segments of data, instead of each segment having its own register. The corresponding algorithm for this is shown in Alg.~\ref{alg:registerreuse}, and the circuit in Fig.~\ref{fig:condensedcircuit}. Note that in this circuit design, initialisation back to $\ket{0}^{\otimes n}$ and preparation through $U_{\mathrm{prep}}^{(i)}$ is done immediately after mid-circuit measurement has occurred. This way, the main registers SI do not spend unnecessary time idle while state preparation on the registers could be done.
\begin{algorithm}[t]
\caption{Q-MERGE Stage 1 (alternative): preparation register reuse with mid-circuit measurement and reset.}
\label{alg:registerreuse}
\begin{algorithmic}[1]
\Require $\tilde{M}<M$ preparation registers, normalised segments
         $\{\ket{\psi_i}\}_{i=0}^{M-1}$, and preparation circuits
         $U_{\mathrm{prep}}^{(i)}:\ket{0}^{\otimes n}\mapsto\ket{\psi_i}$.
\Ensure Weighted stack $\ket{\psi_\textbf{w}}$ on $S\otimes I$, with measurement
        outcomes recorded as each preparation register is reused.
\Statex
\State Initialise $S\gets\ket{+}^{\otimes n}$ and $I\gets\ket{+}^{\otimes m}$.
\For{$i=0$ \textbf{to} $M-1$ \textbf{pipelined}}
    \State Select an available preparation register $A_j$, where
           $j=i\bmod\tilde{M}$.
    \State Apply $U_{\mathrm{prep}}^{(i)}$ to $A_j$.
    \State Apply $C^{m}\mathrm{SWAP}_{n}$ between $A_j$ and $S$,
           controlled on $\ket{i}_I$.
    \State Measure $A_j$ in the computational basis and record the outcome $u_i$.
    \State Reset $A_j\gets\ket{0}^{\otimes n}$ for reuse.
\EndFor
\State \Return $\ket{\psi_\textbf{w}}$ on $S\otimes I$ and the measurement
       record $\vec{u}=(u_0,\ldots,u_{M-1})$.
\end{algorithmic}
\end{algorithm}

Depending on the number of qubits available, $\tilde{M}$ would be chosen such that this process can be done to ensure there is no wait time between the main register and the preparation registers. The state preparation framework would be pipelined, where segments $\ket{\psi_i}$ of the target state $\ket{\psi}$ are all simultaneously being prepared while one $\ket{\psi_i}$ is being swapped into the main registers until the full state is prepared. This achieves the exact same $\mathcal{O}(M\log{(m)}^3+Mn)$ circuit depth as in Alg.~\ref{alg:stage1}, with no impact to the operation or run-time on the main registers.

\begin{figure}
    \noindent
    \hspace*{-1cm}%
    \includegraphics[scale=0.82]{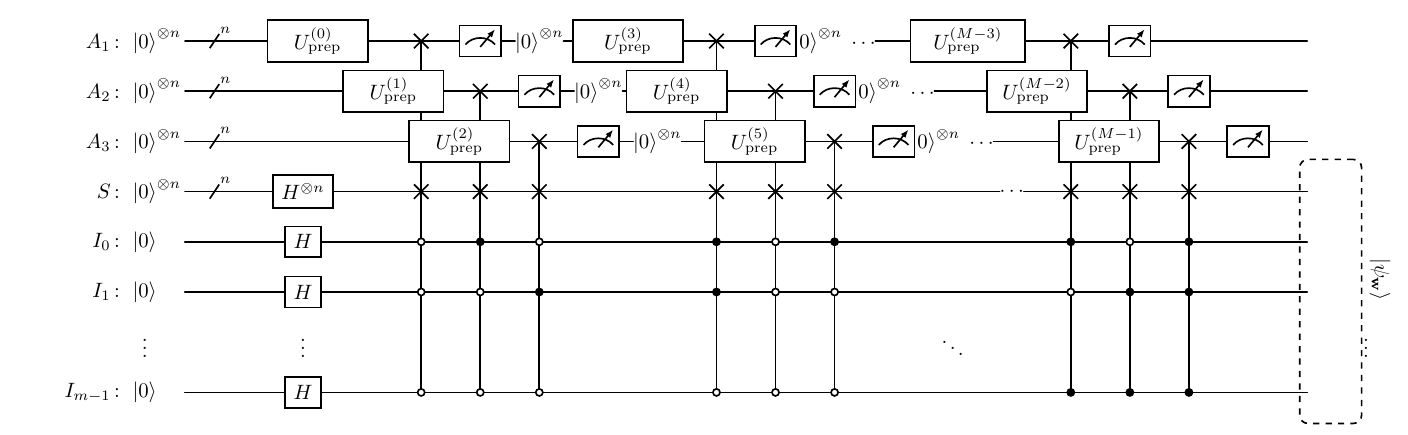}
    \caption{Pipelined Q-MERGE circuit with a fixed number, $\tilde{M}=3$ , of preparation registers.}
    \label{fig:condensedcircuit}
\end{figure}

\section{Results}
\subsection{Numerical Benchmarking}
\label{sec:numerical results}
\begin{figure}[t]
    \centering
    \includegraphics[width=1\linewidth]{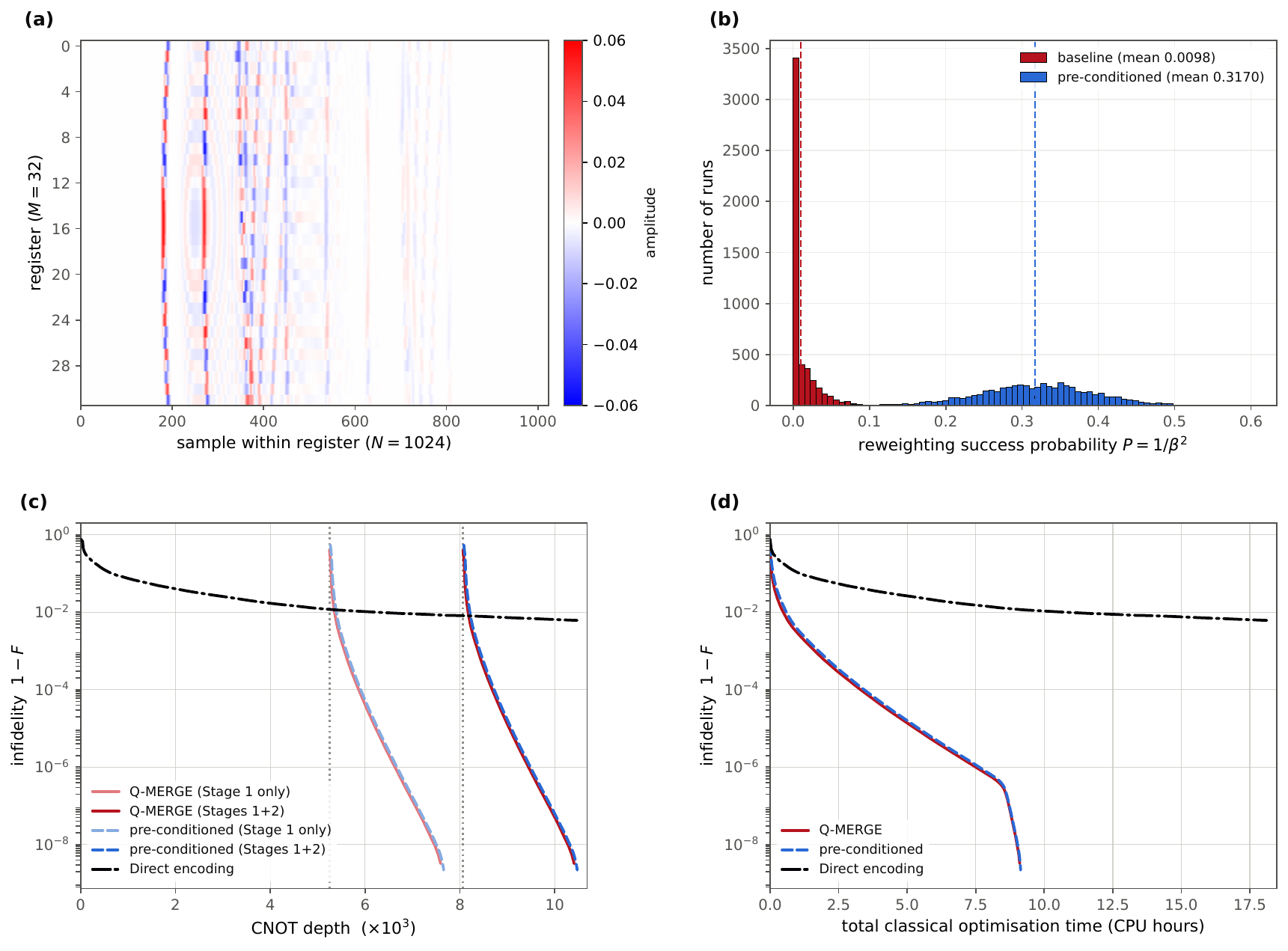}
    \caption{\textbf{(a)} The ultrasound dataset. \textbf{(b)} Distribution of the Stage-2 reweighting success probability of Q-MERGE $P = 1/\beta^2$ over 5000 simulated Stage-1 measurement outcomes for the ultrasound data. Three brickwall layers of preconditioning raise the mean from $P = 0.0098$ to $0.317$, for an added CNOT depth of only 18 gates. Dashed lines mark the means. \textbf{(c)} End-to-end CNOT circuit depth from encoding the ultrasound data with Q-MERGE, preconditioned Q-MERGE and direct encoding, using the same underlying segment-level preparation (SSO). Both the depth of the Q-MERGE (Stage 1 only) and the full Q-MERGE procedure are shown. Each curve is traced out by increasing the number of SSO layers: applied to every segment in parallel for Q-MERGE, and to the single $n+m = 15$ qubit circuit for direct encoding. Q-MERGE carries a fixed depth cost (vertical grey dotted line), on top of which the segment-level preparation depth accumulates. The 3 brickwall preconditioning layers add only a depth of 18 CNOT to this fixed cost and are not resolvable on this scale. \textbf{(d)} Total classical optimisation time to generate the state-preparation circuit.}
    \label{fig:numresults}
\end{figure}
We now show the end-to-end process of using Q-MERGE to encode a sample of ultrasound data. This target state forms a 2D array of size $32 \times 1024$, and represents the type of data structures needed to be prepared for quantum-assisted medical imaging applications, with the dataset shown in Fig. \ref{fig:numresults}.

As laid out in Sec. \ref{sec:stage1}, a choice was made in partitioning the 2D array. A natural choice for this is partitioning the data up into its $M=32$ rows, each of length $N=1024$, as each row represents smooth time series data from individual ultrasound detection elements. Each row was prepared using the matrix product state preparation method of \citeauthor{green2026quantum} \cite{green2026quantum} known as \textit{Schmidt Spectrum Optimisation} (SSO), with $200$ L-BFGS-B iterations per staircase layer of SU(4) gates. The aggregate fidelity is computed from the segment level fidelity by Eq.~\eqref{eq:multiregister fidelity}.

This approach achieved an infidelity of $3.5 \times10^{-9}$, more than six orders of magnitude better than directly encoding the target state $\ket{\psi}$ on the SI registers under the same method, for the same CNOT circuit depth of $1.04\times 10^4$. The infidelity of direct encoding and Q-MERGE with SSO is shown in Fig. \ref{fig:numresults}, compared to the CNOT circuit depth and classical optimisation time to generate the SSO layers. Note that the per-layer optimisation cost falls as the encoding improves, as high fidelity states produced by SSO are able to be represented by an MPS with low bond dimension. The CNOT circuit depth includes the full SELECT-SWAP network and reweighting operations, and the optimisation time quoted represents time on a single consumer laptop (16-core Intel Core Ultra X7, 32 GB RAM), summed over all $M$ segments fitted sequentially. In practice the classical optimisation could be done in parallel as each segment level preparation is independent, and use of super computing resources would reduce the optimisation time from hours to minutes.

Additionally, Stage 2 of Q-MERGE was simulated: the measurement of the circuit's preparation registers was classically simulated as described in Sec. \ref{sec:classical-sim} in order to produce the weighted stack state of the preparation states, from which the success probability of the reweighting can be calculated. As noted in \ref{sec:classical preconditioning}, states in which each segment of the state contains a large number of small amplitudes may lead to ill-conditioned weights $w_i$, and a prohibitively small success probability for reconstructing the original state vector. This is indeed the case with this target state, leading to a reweighting success probability of only $P=$ 0.0098, before the preconditioning step was applied. The classical preconditioning step $V$ boosted this to $P=0.3170$, at the cost of only 3 brickwall layers of SU(4) gates, highlighting the importance of the classical preconditioning in recovery of the desired state. The introduction of the preconditioning did not degrade the infidelity of the state preparation, with the preconditioned state achieving an infidelity of $2.2\times10^{-9}$.

\subsection{Quantinuum System Model H2 Hardware Implementation}
\label{sec:hardware}
We implemented stage 1 of Q-MERGE on the Quantinuum System Model H2 trapped-ion quantum computer. We encode $M=4$ rows of length $N=2^{10}=1024$ from the ultrasound data, creating a circuit on $53$ qubits, (using $4$ preparation registers of $10$ qubits, a storage register of $10$ qubits, $2$ index qubits and $1$ ancilla). Each $U_{\mathrm{prep}}^{(i)}$ is intentionally kept shallow, at 1 layer of SSO state preparation, so that the total gate count and noise contribution for this hardware experiment consists predominantly of the Q-MERGE protocol rather than the underlying state preparation subroutine. The goal of this experiment is to verify that Q-MERGE has successfully merged the $4$ individually prepared $10$ qubit segments into the desired weighted stack on $12$ qubits of the index and storage register. To quantify the fidelity of the prepared weighted stack to the desired weighted stack, a quantum state certification algorithm known as \textit{shadow overlap certification} \cite{Huang2024-jo} was used. However, since each shot prepares a weighted stack with different randomised weights, we cannot certify the whole weighted stack directly. Instead, we apply the shadow overlap technique on each branch of the weighted stack individually. The weighted stack for our state is 
\begin{equation}
  \ket{\tilde\psi_\text{w}}= \sum_{i=0}^{3}w_i  \ket{\tilde\psi_i}_S \otimes \ket{i}_I
\end{equation}
Measuring the index register as $\ket{i}_I$ will result in the storage register S collapsing into one of the $4$ different $\ket{\tilde{\psi}_i}$. We are able to determine that the compressed state on the storage and index register contains all the information (up to scale) of the segments by certifying that the correct $\ket{\tilde{\psi}_i}$ is prepared on the storage register. The certification procedure works by estimating an efficiently measurable local observable that acts as a proxy for the fidelity known as the shadow overlap $\omega$. The algorithm for calculating the shadow overlap is explained in detail in Sec. \ref{sec:shadow overlap}. Fig. \ref{fig:hardware_results_52q} shows the resulting histograms displaying the shadow overlaps obtained from the theoretical noise-free simulation and the H2-1 quantum hardware. 

From this we form an estimate $\tilde F_i$ of the fidelity to the desired $\ket{\tilde{\psi}_i}$ through inverse linear regression. We uniformly sample fidelities between 0 and 1, construct a random state at each, and compute its shadow overlap exactly, to construct a calibration curve, over which the shadow overlap and fidelity are strongly correlated (Pearson $r = 0.9996$). This process is given in more detail in Sec. \ref{sec:inverse-regression}. The results are shown in Fig. \ref{fig:hardware_results_52q}, giving the inferred fidelity for each state as shown in Tab. \ref{tab:shadow_overlap_52q}.

\begin{table}[ht]
    \centering
    \caption{Per-branch shadow overlap and estimated fidelity $\tilde F_i$ for the four
encoded rows to $\ket{\tilde\psi_i}$ on the Quantinuum System Model H2, from $450$ shots
($150$ each of $X$, $Y$, $Z$) conditioned on the index register, for the upscaled ($N=2^{10}$) rows of ultrasound data.}
    \label{tab:shadow_overlap_52q}
    \begin{tabular}{c c c c}
        \toprule
        Row $i$ & Measured $\omega$ & Fidelity $F$ & $90\%$ interval \\
        \midrule
        0 & $0.795 \pm 0.076$ & $0.589$ & $[0.339,\,0.839]$ \\
        1 & $0.895 \pm 0.068$ & $0.790$ & $[0.567,\,1.000]$ \\
        2 & $0.860 \pm 0.075$ & $0.721$ & $[0.474,\,0.967]$ \\
        3 & $0.860 \pm 0.071$ & $0.720$ & $[0.492,\,0.949]$ \\
        \bottomrule
    \end{tabular}
\end{table}

The weighted average estimated fidelity $\tilde F=\sum_i\sigma_i^2\tilde F_i$, with respect to the weights of the original data,
\begin{equation}
    \sigma_i^2=\{0.210194,\,0.296888,\,0.249177,\,0.243741\}, 
\end{equation}
is $\tilde F = 0.7135\pm0.0720$, demonstrating the viability of Q-MERGE on the Quantinuum System Model H2.

\begin{figure}[H]
    \centering
    \includegraphics[width=1\linewidth]%
    {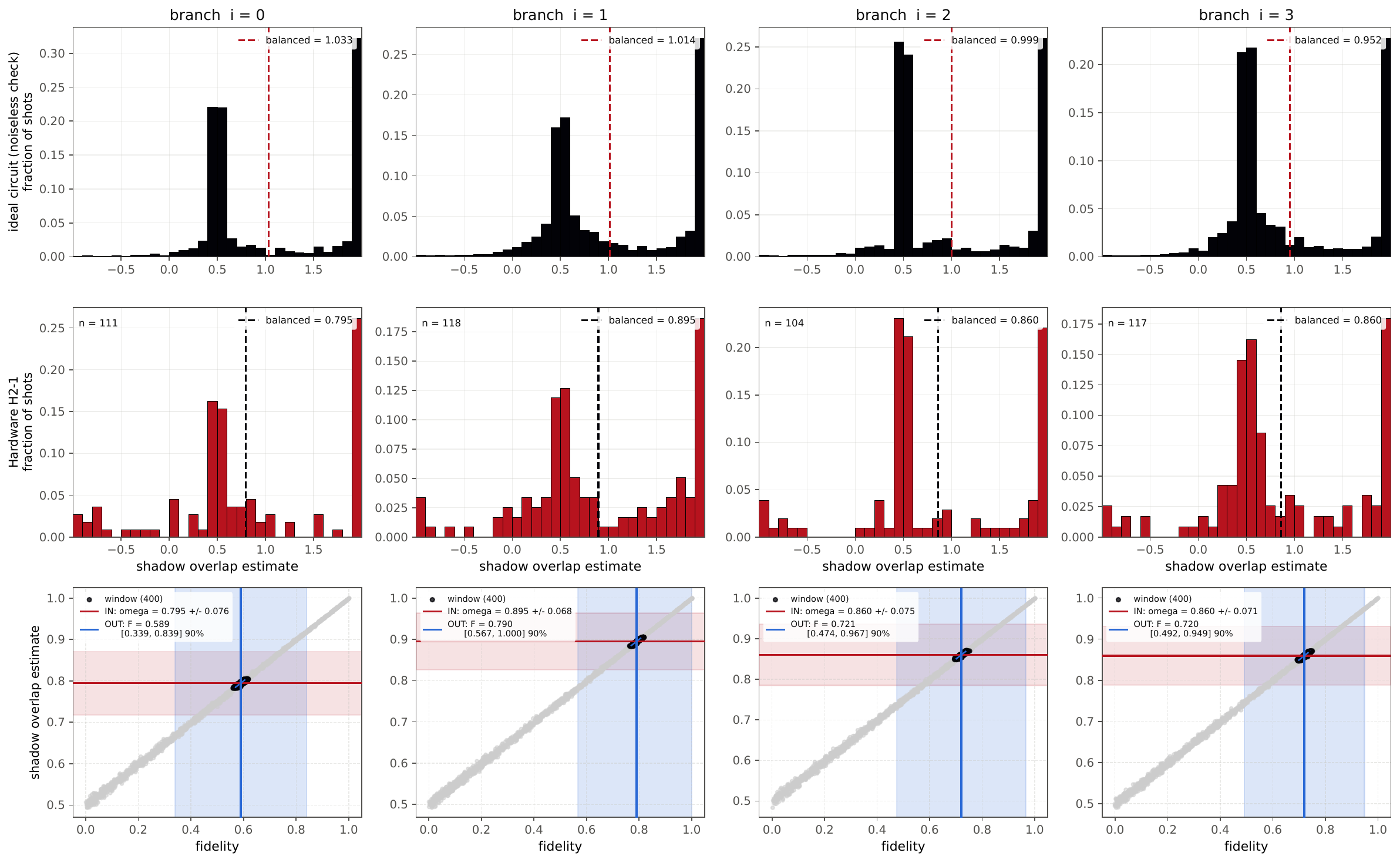}
    \caption{Shadow overlap results for the four $\ket{\psi_i}$ branches for the upsampled ultrasound data: noiseless circuit simulation vs.\ Quantinuum hardware with a 53 qubit circuit comprised of 954 1-qubit gates and 478 2-qubit gates, with fidelity estimates via inverse linear regression (Table~\ref{tab:shadow_overlap_52q}).}
    \label{fig:hardware_results_52q}
\end{figure}

\subsection{Scaling for Haar random states}
It is known that the application of random \textit{brickwall} unitaries causes the resultant state to quickly converge to Haar random \cite{Haferkamp2022-jd}. Therefore, the scaling of the success probability of Q-MERGE on Haar random states can serve as an estimate for a lower bound of the success probability on preconditioned states. It can be interpreted as a special case of our method using zero optimisation steps during preconditioning, initialising each brickwall layer with Haar random SU(4) gates. We demonstrate that this success probability scales independently of $N$, the size of each segment $\ket{\psi_i}$, and only slowly decays with $M$, the total number of segments. These numerical scaling relations are shown in Fig. \ref{fig:haarscaling}, where even for the massive $M=10^7$, the success probability remains above $P=5 \times 10^{-2}$. 
\begin{figure}
    \centering
    \includegraphics[width=0.6\linewidth]{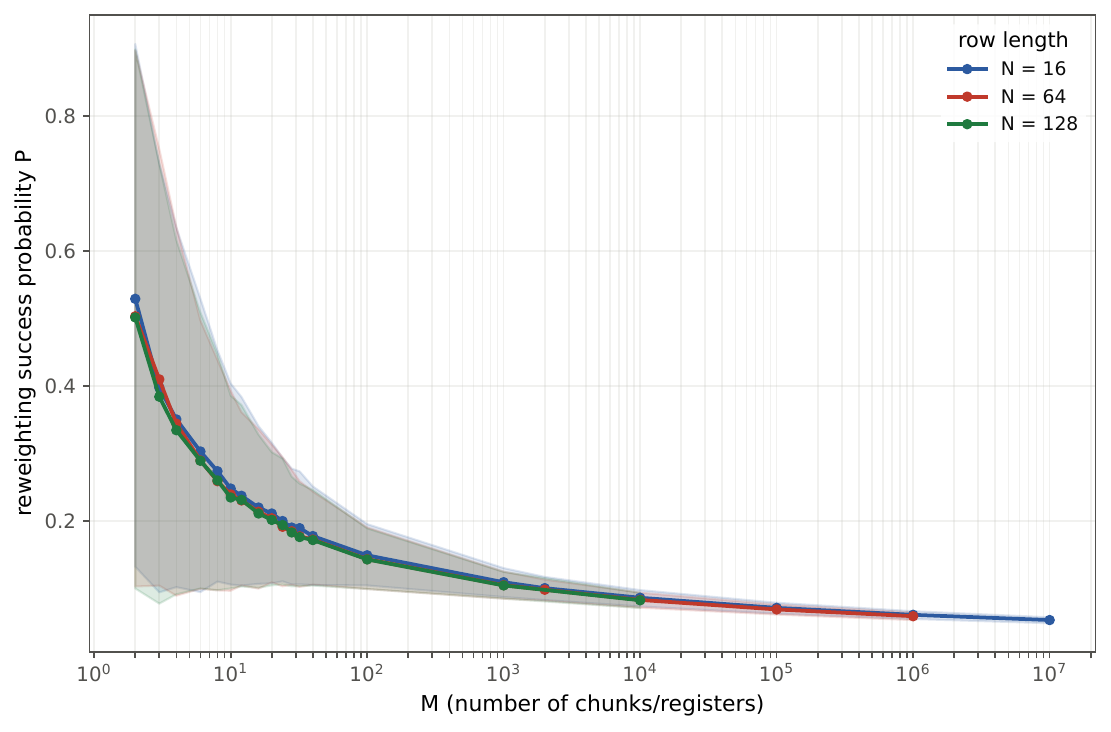}
    \caption{Scaling of the reweighting success probability $P=1/\beta^2$ for increasing $M$, with different $N$ values.}
    \label{fig:haarscaling}
\end{figure}
\section{Conclusion and Future Work}
We have presented Q-MERGE, a two-stage algorithm that provides a scalable approach to the state preparation of massive classical datasets. By partitioning large classical data into smaller states $\ket{\psi_i}$, we avoid the difficulties associated with direct preparation of large classical data, as well as allowing for the preparation of each segment to be done in parallel. Alternatively, each of these classical segments can be loaded in ``on the fly'', of interest when either the quantum hardware supports only a fixed number of qubits, or when the data to be encoded is too large to be stored classically. We believe this method provides a fundamental building block to state preparation methods, with the potential of improving the scalability of all existing state preparation approaches.

The issue of the success probability corresponding to the non-unitary reweighting step was addressed by a classical preconditioning layer on the data. In the case of the ultrasound data we show this is able to boost the success probability from $0.0098$ to $0.3170$ with only 3 layers of \textit{brickwall} SU(4) gates, only inducing a slight fixed cost to each segment's state preparation depth. The method is also demonstrated on the Quantinuum System Model H2 physical qubits, demonstrating encoding a subset of the ultrasound data to a weighted average estimated fidelity of $\tilde F=0.7135\pm0.0720$, which could be further improved by incorporating error correction codes to create logical qubits \cite{Self2022-fn}.

We also numerically show that Q-MERGE is able to scale to a large number of segments $M$, as well as independence of the size of each segment $N$, maintaining reasonable success probabilities for segmenting Haar random states into $M=10^7$ segments. For cases where the classical data is so large that preconditioning is intractable, a randomly generated brickwall circuit can still be used to move the reweighting success probability into the favourable regime of the Haar random states.

Future work into this method would involve investigating how effective the classical preconditioning layer is for different types of datasets, as well as different choices of ansatzes for $V$ and different choices of the loss function for IPR. Different initial states of the storage register could be investigated which in combination with classical preconditioning could further increase the reweighting success probability.

Furthermore, our method hints at a potential paradigm for dedicated state preparation hardware, where a main quantum processor is supported by $\tilde{M}$ purposely designed state preparation registers that exist purely for the purpose of efficient state preparation, while the main quantum processor loads in these states before performing the quantum algorithm. This is shown in Fig. \ref{fig:condensedcircuit}. Each preparation register of $n$ qubits only needs to be internally connected up to the state preparation requirements, to which the SSO or MPS framework may only require nearest-neighbour connectivity. Each preparation register will then only need to perform the same exact CSWAP with the main quantum processor, which the hardware can be specifically designed to do so as to accommodate high fidelity.

\section{Acknowledgements}
We thank Quantinuum for providing hardware access and recognising our team as a finalist in their 2026 SG Grand Challenge. This project was funded through the Department of Defence Advanced Strategic Capabilities Accelerator (ASCA).

\bibliographystyle{apsrev4-2}
\bibliography{references}
\appendix

\section{Diagonal block encoding via controlled rotations}
\label{sec:block encoding}
Let $\beta = \max \{|\alpha_i|:i=0,\cdots,M-1\}$, and define rotation angles
\begin{equation}
  \theta_i = \arccos\left(\frac{|\alpha_i|}{\beta}\right) \in \Bigl[0,\tfrac{\pi}{2}\Bigr],
  \qquad i = 0, \dots, M-1,
  \label{eq:theta}
\end{equation}
which are well-defined since $|\alpha_i|/\beta \in [0,1]$ for every $i$. Apply the SELECT-rotate unitary
\begin{equation}
  U = \sum_{i=0}^{M-1} |i\rangle\langle i|_I \otimes R_i,
  \label{eq:U-reweight}
\end{equation}
where 
\begin{equation}
R_i =
\begin{pmatrix}
\dfrac{\alpha_i}{|\alpha_i|}\cos\theta_i &
-\dfrac{\alpha_i}{|\alpha_i|}\sin\theta_i
\\[1ex]
\sin\theta_i &
\cos\theta_i
\end{pmatrix}.
\end{equation}
Acting $U$ on $|\psi_\textbf{w}\rangle \otimes |0\rangle_Q$:
\begin{align}
  U\bigl(|\psi_\textbf{w}\rangle |0\rangle_Q\bigr)
  &= 
     \sum_{i=0}^{M-1} w_i \, \ket{\psi_i}_S |i\rangle_I
     \Bigl(
       \frac{\alpha_i}{\beta}|0\rangle_Q
       + \sqrt{1 - \frac{|\alpha_i|^2}{\beta^2}}\,|1\rangle_Q
     \Bigr) \nonumber \\
  &= \frac{1}{\beta}
     \sum_{i=0}^{M-1} \sigma_i\ket{\psi_i}_S |i\rangle_I |0\rangle_Q
     \;+\;
     \sum_{i=0}^{M-1}
       w_i\sqrt{1 - \frac{|\alpha_i|^2}{\beta^2}}\;
       \ket{\psi_i}_S |i\rangle_I |1\rangle_Q. 
  \label{eq:after-rotation}
\end{align}
The $|0\rangle_Q$ term has the correct weighting $\sigma_i$ on every branch as
the factors $w_i$ have cancelled exactly, giving us the desired state $\ket{\psi}$ once normalised, and the $|1\rangle_Q$
term collects the orthogonal junk state. Now we measure $Q$ in the computational basis, with postselection on $\ket{0}$ giving probability of success as 
\begin{equation}
    P(\ket{0})=\frac{1}{\beta^2}\sum_{i=0}^{M-1}\sigma_i^2=\frac{1}{\beta^2}
\end{equation}
\section{Minimum IPR and guaranteed reweighting success for phase states}
\label{sec:minimum IPR}

We now show that when the target state $\ket{\psi}$ has perfectly flat amplitudes, known as a \emph{phase state}, the state simultaneously saturates the lower bound of the IPR in Eq.~\eqref{eq:ipr} and yields $\beta = 1$ in Algorithm~\ref{alg:stage2}, i.e.\ the weighting
redistribution step of Stage 2 succeeds with probability exactly $1$.

\subsection{Setup}
As in Sec.~\ref{sec:stage1}, write the target state on the $MN$-dimensional computational basis as
\begin{equation}
  \ket{\psi} = \sum_{i=0}^{M-1}\sum_{k=0}^{N-1} (\psi)_{iN+k} \ket{k}\otimes\ket{i}
  = \sum_{i=0}^{M-1} \sigma_i \ket{\psi_i}\otimes\ket{i},
  \qquad \ket{\psi_i} = \sum_{k=0}^{N-1} (\psi_i)_k \ket{k} \in \mathbb{C}^N,
\end{equation}
with $\sigma_i (\psi_i)_k = (\psi)_{iN+k}$. A phase state is one for which every amplitude has the same magnitude,
\begin{equation}
  \ket{\psi} = \frac{1}{\sqrt{MN}}\sum_{i=0}^{M-1}\sum_{k=0}^{N-1} e^{i\phi_i(k)}\ket{k}\otimes\ket{i},
  \qquad (\psi)_{iN+k} = \frac{1}{\sqrt{MN}}\, e^{i\phi_i(k)},
  \label{eq:phase-state-def}
\end{equation}
for arbitrary real phases $\phi_i(k)$, so that $|(\psi)_{iN+k}| = 1/\sqrt{MN}$ for every $i,k$. Comparing Eq.~\eqref{eq:phase-state-def} to the segment decomposition, each segment itself has flat amplitudes,
\begin{equation}
  \ket{\psi_i} = \frac{1}{\sqrt{N}}\sum_{k=0}^{N-1} e^{i\phi_i(k)}\ket{k},
  \qquad \sigma_i = \frac{1}{\sqrt{M}} \quad \text{for every } i,
  \label{eq:phase-state-segments}
\end{equation}
where the equal weighting $\sigma_i = 1/\sqrt{M}$ follows from normalisation of $\ket{\psi_i}$
together with the requirement $|\sigma_i (\psi_i)_k| = 1/\sqrt{MN}$ for all $k$.

\subsection{Minimum IPR}
Substituting Eq.~\eqref{eq:phase-state-def} into the definition of the IPR, Eq.~\eqref{eq:ipr}, with
$D = MN$,
\begin{equation}
  \mathrm{IPR}(\psi) = \sum_{i=0}^{M-1}\sum_{k=0}^{N-1} \left|(\psi)_{iN+k}\right|^4
  = MN \cdot \left(\frac{1}{\sqrt{MN}}\right)^{4}
  = \frac{1}{MN},
\end{equation}
which saturates the lower bound $1/D$ stated below Eq.~\eqref{eq:ipr}. Thus a phase state is exactly the flattest possible amplitude distribution on $MN$ basis states, in the IPR sense used by the flattening layer $V$ of Sec.~\ref{sec:classical preconditioning}.

\subsection{Guaranteed reweighting success}
Let $\vec u = (u_0,\dots,u_{M-1})$ be any outcome of the Stage 1 measurement (Algorithm~\ref{alg:stage1}). Since every segment in Eq.~\eqref{eq:phase-state-segments} has flat amplitude, $\left|(\psi_k)_{u_k}\right| = 1/\sqrt{N}$ regardless of $k$ or the sampled value $u_k$, so the unnormalised weight
$\tilde w_i = \prod_{k\neq i}(\psi_k)_{u_k}$ has magnitude
\begin{equation}
  \left|\tilde w_i\right| = \prod_{k\neq i} \left|(\psi_k)_{u_k}\right| = N^{-(M-1)/2}
  \qquad \text{for every } i,
  \label{eq:phase-state-wtilde}
\end{equation}
independent of $i$ and of the measurement outcome $\vec u$, so only the phase of $\tilde w_i$ is different between branches. Consequently $\mathcal{N} = \sum_i |\tilde w_i|^2 = M N^{-(M-1)}$, so the normalised weights of Eq.~\eqref{eq:output_after_measure_un} also all have the same magnitude,
\begin{equation}
  |w_i| = \frac{|\tilde w_i|}{\sqrt{\mathcal{N}}} = \frac{1}{\sqrt{M}}
  \qquad \text{for every } i.
  \label{eq:phase-state-w}
\end{equation}
Using $\alpha_i = \sigma_i/w_i$ from Sec.~\ref{sec:classical preconditioning} together with Eqs.~\eqref{eq:phase-state-segments} and \eqref{eq:phase-state-w},
\begin{equation}
  |\alpha_i| = \frac{\sigma_i}{|w_i|} = \frac{1/\sqrt{M}}{1/\sqrt{M}} = 1
  \qquad \text{for every } i,
\end{equation}
so $\beta = \max_i |\alpha_i| = 1$, and the reweighting success probability of Sec.~\ref{sec:classical preconditioning} is
\begin{equation}
  P(\ket{0}_Q) = \frac{1}{\beta^2} = 1.
\end{equation}
That is, for a phase state the reweighting step of Stage 2 succeeds \emph{deterministically}, for every measurement outcome $\vec u$ that can be obtained in Stage 1: no post-selection failure is possible.
\section{Decomposition of SELECT-SWAP gates}
\label{sec:swapdecomp}
A SELECT-SWAP operation can be decomposed into 2 $\mathrm{C}^m\mathrm{NOT}$ gates, and $n$ single qubit controlled CSWAP operations with the addition of a single zeroed ancilla qubit. To make clear the distinction between the types of ancilla qubits, a zeroed ancilla qubit is one where it begins in the $\ket{0}$ state and is restored to $\ket{0}$ afterwards, and a ``borrowed'' ancilla qubit is one which can begin in any state and is restored afterwards. \citeauthor{Claudon2024-ry} \cite{Claudon2024-ry} show that these multi-controlled $\mathrm{C}^m\mathrm{NOT}$ gates can be efficiently decomposed with a circuit depth of $\mathcal{O}(\log{(m)}^3)$ single qubit and CNOT gates and size $\mathcal{O}(m \log{(m)}^4)$, using a single borrowed ancilla qubit from the storage register. This $\mathrm{C}^m\mathrm{NOT}$ acts sequentially on a zeroed qubit from which $n$ single register CSWAP are conditioned on to swap the state into the storage register $S$. The $\mathrm{C}^m\mathrm{NOT}$ operation is undone so that the zeroed ancilla may be reused in the other $\mathrm{C}^m\mathrm{NOT}$ required operations. There exist efficient decompositions of CSWAP gates into a fixed number of single qubit and CNOT gates, leading to a total circuit depth of $\mathcal{O}(M\log{(m)}^3+Mn)$ single qubit and CNOT gates for the whole process on the main registers.
\section{Classical simulation of the compression measurement}
\label{sec:classical-sim}
The numerical results in Sec.~\ref{sec:numerical results} require repeatedly sampling outcomes of the joint measurement on the preparation registers $A_0,\dots,A_{M-1}$. Naively, this requires sampling from the full distribution
\begin{equation}
  P(\vec{u}) = \frac{1}{MN}\sum_{i=0}^{M-1}\prod_{k\neq i}\big|(\psi_k)_{u_k}\big|^2,
  \qquad \vec{u}\in\{0,\dots,N-1\}^M,
  \label{eq:full-dist}
\end{equation}
whose support has size $N^M$, making direct sampling intractable beyond small $M$.

\begin{proposition}\label{prop:mixture}
Define, for each $i \in \{0,\dots,M-1\}$,
\begin{equation}
  Q_i(\vec{u}) = \frac{1}{N}\prod_{k\neq i}\big|(\psi_k)_{u_k}\big|^2.
  \label{eq:Qi}
\end{equation}
Then $Q_i$ is a valid probability distribution over $\vec{u}$, and
\begin{equation}
  P(\vec{u}) = \frac{1}{M}\sum_{i=0}^{M-1} Q_i(\vec{u}).
  \label{eq:mixture}
\end{equation}
\end{proposition}

\begin{proof}
We have $Q_i(\vec{u}) = |\langle \vec{u} \,|\, \phi_i\rangle|^2$, where
$\ket{\phi_i} = \ket{+}_{A_i}\otimes\bigotimes_{k\neq i}\ket{\psi_k}_{A_k}$
is a tensor product of normalised states, and is therefore itself normalised. Since $\{\ket{\vec{u}}\}$ is an orthonormal basis,
\begin{equation}
  \sum_{\vec{u}} Q_i(\vec{u}) = \langle\phi_i|\phi_i\rangle = 1,
\end{equation}
so $Q_i$ is a valid distribution. The identity \eqref{eq:mixture} follows by direct comparison of \eqref{eq:full-dist} and \eqref{eq:Qi}.
\end{proof}

Proposition~\ref{prop:mixture} shows $P$ is a mixture of $M$ distributions, each of which factorises across registers: register $i$ is uniform on $\{0,\dots,N-1\}$, while every other register $k$ follows the true Born distribution $|(\psi_k)_{u_k}|^2$ of its segment. This yields an efficient classical sampling algorithm:
\begin{enumerate}
  \item Draw $i$ uniformly from $\{0,\dots,M-1\}$.
  \item Draw $u_i$ uniformly from $\{0,\dots,N-1\}$.
  \item For each $k \neq i$, draw $u_k$ independently from the categorical distribution $\{|(\psi_k)_{u_k}|^2\}_{u_k=0}^{N-1}$.
\end{enumerate}
By the law of total probability, the resulting $\vec{u}$ is exactly distributed according to $P$. Precomputing the cumulative distribution of each segment once ($\mathcal{O}(N)$ per segment) allows each categorical draw in step 3 to be performed in $\mathcal{O}(\log N)$ via inverse transform sampling, giving a total cost of $\mathcal{O}(M\log N)$ per sample of the compression measurement outcome, independent of the size $N^M$ of the full outcome space. This procedure was used to generate the measurement outcomes underlying the numerical results in Sec.~\ref{sec:numerical results}.

\section{Shadow Overlap Certification}
\label{sec:shadow overlap}
\textcite{Huang2024-jo} propose a method of certifying the fidelity of an unknown $n$-qubit state $\rho$ with a target state $\ket{\psi}$. Using only single-qubit measurements on $\rho$, a quantity known as the shadow overlap $\omega$ is computed for each measurement taken on the circuit. This is done by first uniformly selecting a random $k \in \{0,..,n-1\}$. All qubits are then measured in the $Z$-basis (computational basis) apart from the $k$'th qubit which is measured with equal probability in either the $X,$ $Y,$ or $Z$ Pauli basis. The measurement results of all but the $k$'th qubit are stored as a bitstring $z \in \{0,1\}^{n-1}$, and the measured state of the $k$'th qubit in its respective basis is stored as $\ket{s}$. We then classically reconstruct a model of the target state's $k$'th qubit, conditional on the measurement $z$, as
\begin{equation}
\ket{\Psi_{k,z}}:= \frac{1}{\sqrt{|\Psi(z^{(0)})|^2+|\Psi(z^{(1)})|^2}}\left[\Psi(z^{(0)})\ket{0}+\Psi(z^{(1)})\ket{1}\right]
\end{equation}
where $z^{(a)}$ is the binary string that equals $a \in \{0,1\}$ on the $k$'th qubit and $z \in \{0,1\}^{(n-1)}$ in the remaining $n-1$ qubits. Here $\Psi(x)=\braket{x}{\psi}$ denotes the amplitude of the target state on the computational basis state $\ket{x}$, $x\in\{0,1\}^n$, so $\ket{\Psi_{k,z}}$ is the target's $k$'th qubit conditioned on the other $n-1$ qubits being found in the bitstring $z$. The shadow overlap is then computed using the observed state $\ket{s}$ of the $k$'th qubit as 
\begin{equation}
\omega := \langle \Psi_{k,z} \rvert \, (3 \lvert s\rangle\langle s\rvert - \mathbb{I}) \, \lvert \Psi_{k,z} \rangle.
\end{equation}
In practice, $M$ multiple measurements of a circuit are made to compute many shadow overlaps generating an estimate $\bar{\omega}=\frac{1}{M}\sum_{i}\omega_i$ of the true value $ \mathbb{E}[\omega] \approx\bar{\omega} $. This shadow overlap, $\mathbb{E}[\omega]$, acts as a surrogate for the fidelity.

\section{Inverse regression for fidelity estimation}
\label{sec:inverse-regression}

For each segment we generate $T=3000$ calibration states of prescribed fidelity
$f\sim\mathrm{Unif}[0,1]$,
\begin{equation}
  \ket{\chi}=\sqrt{f}\,\ket{\psi_i}+\sqrt{1-f}\,\ket{\nu},
  \qquad \braket{\psi_i}{\nu}=0,
\end{equation}
where $\ket{\nu}$ is a complex Gaussian vector projected onto the orthogonal
complement and renormalised, so $\ket{\chi}$ is automatically normalised and
$F=|\braket{\chi}{\psi_i}|^{2}=f$ exactly. For each of these states the exact value of the expected shadow overlap $\mathbb{E}(\omega)$ can be computed exactly due to the small number $n$ of qubits on the storage and preparation registers. The resultant discrepancy between the shadow overlap values and the true fidelity is purely emergent from the fact that different states of the same shadow overlap will result in a range of fidelities which the calibration curve is used to capture. Over the calibration states the Pearson correlation between $\omega$ and $F$ is $r = 0.9996$.

After generating an estimate of the shadow overlap $\omega^\star$ from the hardware, the closest 400 points to $\omega^\star$, measured by $|\omega^\star-\omega|$ are fit by ordinary least squares to
\begin{equation}
  F=a+b\,(\omega-\omega^\star)+\varepsilon,
  \qquad \tilde F=\hat a,
\end{equation}
providing a local linear estimate of the shadow overlap to fidelity relation.

The standard error $s_\omega$ around the estimate $\omega^\star$ combines with the uncertainty in the range of fidelities corresponding to a single fixed $\omega$. As the residuals from the local linear fit are skewed, bootstrapping is used to produce a confidence interval on $F$. We draw $2\times10^{5}$ samples $\tilde\omega^{(t)}\sim\mathcal{N}(\omega^\star,s_\omega)$, resample a residual
$\varepsilon^{(t)}$ from $\{r_j\}$ with replacement, and evaluate
\begin{equation}
  \tilde F^{(t)}=\hat a+\hat b\,(\tilde\omega^{(t)}-\omega^\star)+\varepsilon^{(t)},
\end{equation}
clipped to $[0,1]$. The reported interval is the corresponding empirical
percentile range of $\{\tilde F^{(t)}\}$.
\end{document}